\def\version{8 January 2015}

\documentclass[reqno]{amsart}
\usepackage{amssymb,graphics,graphicx,bbm,hyperref,color}
\usepackage{mathtools,MnSymbol}

\usepackage{multicol}
\usepackage{enumitem}

\definecolor{gray}{rgb}{0.93,0.93,0.93}
\definecolor{light-gold}{rgb}{0.99,0.97,0.78}

\def\be{\begin{equation}}
\def\ee{\end{equation}}
\def\bm{\begin{multline}}
\def\eem{\end{multline}}
\def\bfig{\begin{figure}[htb]}
\def\efig{\end{figure}}
\newcommand{\dd}{{\rm d}}
\newcommand{\e}[1]{\,{\rm e}^{#1}\,}
\newcommand{\ii}{{\rm i}}
\def\Tr{{\operatorname{Tr\,}}}

\renewcommand{\Re}{{\rm Re}\;}

\numberwithin{equation}{section}
\newtheorem{theorem}{Theorem}[section]

\newtheorem{lemma}[theorem]{Lemma}
\newtheorem{corollary}[theorem]{Corollary}

\newcommand{\caE}{{\mathcal E}}
\newcommand{\caH}{{\mathcal H}}
\newcommand{\caL}{{\mathcal L}}
\newcommand{\bbC}{{\mathbb C}}
\newcommand{\bbE}{{\mathbb E}}
\newcommand{\bbN}{{\mathbb N}}
\newcommand{\bbP}{{\mathbb P}}

\newcommand{\bbZ}{{\mathbb Z}}
\newcommand{\bsh}{{\boldsymbol h}}

\newcommand{\eps}{{\varepsilon}}

\newcommand{\EE}{\mathbb{E}} 
\newcommand{\bE}{\mathbf{E}}
\newcommand{\bP}{\mathbf{P}}
\newcommand{\PP}{\mathbb{P}}
\newcommand{\ZZ}{\mathbb{Z}}
\newcommand{\cA}{\mathcal{A}}
\newcommand{\cB}{\mathcal{B}} 
\newcommand{\cL}{\mathcal{L}}
\newcommand{\cE}{\mathcal{E}}
\renewcommand{\a}{\alpha}
\renewcommand{\b}{\beta}
 
\renewcommand{\d}{\delta} 
\newcommand{\G}{\Gamma}
\newcommand{\g}{\gamma} 
\renewcommand{\L}{\Lambda} 
 
\newcommand{\om}{\omega}
 
\newcommand{\s}{\sigma}

\newcommand{\sm}{\setminus}
\newcommand{\lra}{\leftrightarrow}

\newcommand{\one}{\hbox{\rm 1\kern-.27em I}}

\def\bbone{{\mathchoice {\rm 1\mskip-4mu l} {\rm 1\mskip-4mu l} {\rm 1\mskip-4.5mu l} {\rm 1\mskip-5mu l}}}

\makeatletter
  \def\tagform@#1{\maketag@@@{\footnotesize{(#1)}\@@italiccorr}}
\makeatother

\renewcommand{\eqref}[1]{(\ref{#1})}

\begin{document}

{\hfill\small \version} \vspace{2mm}

\title{Decay of transverse correlations in quantum Heisenberg models}

\author{Jakob E. Bj\"ornberg}
\address{Department of Mathematical Sciences,
Chalmers and University of Gothenburg, 
41296 G\"oteborg, Sweden}
\email{jakob.bjornberg@gmail.com}

\author{Daniel Ueltschi}
\address{Department of Mathematics, University of Warwick,
Coventry, CV4 7AL, United Kingdom}
\email{daniel@ueltschi.org}



\begin{abstract}
We study a class of quantum spin systems
that includes the $S=\tfrac12$ Heisenberg and XY-models,
and prove that two-point correlations exhibit exponential
decay in the presence of a transverse magnetic field.
The field is not necessarily constant, it may be random, and it points in the same direction.
Our proof is entirely probabilistic and it relies on a 
random-loop-representation
of the correlation functions, on stochastic domination, and on first-passage 
percolation.
\end{abstract}

\thanks{\copyright{} 2014 by the authors. This paper may be reproduced, in its
entirety, for non-commercial purposes.\\
The research of \textsc{jeb} is supported by the Knut and Alice
Wallenberg Foundation. We are grateful for mobility support by The Leverhulme Trust through the International Network ÒLaplacians, Random Walks, Quantum Spin SystemsÓ.}

\maketitle

\section{Setting and results}

It has been known since the work of T\'oth~\cite{Toth} and Aizenman
and Nachtergaele~\cite{AN} that certain quantum spin systems may be
represented in terms of a collection of random loops.  The two
representations were recently combined so as to be 
 included in a  larger family
of models~\cite{Uel1}.  We study this family with the addition of
positive transverse fields, and use the loop representation to prove
that two-point correlations decay exponentially.  Some such results can
alternatively be obtained as a consequence of the Lee--Yang theorem
(as we remark below).  However, our method of proof, which uses 
techniques from modern probability theory, is new and
interesting in itself.

This work is one of a growing number of contributions to the
understanding of quantum spin systems using probabilistic graphical
representations.  This includes the recent work by Crawford, Ng and
Starr~\cite{CNS} on emptiness formation in the XXZ model as well as
work on the transverse field Ising 
model~\cite{B-irb,B-van,BG,CI}.
We note in particular that Crawford and Ioffe~\cite{CI}
establish exponential decay of truncated correlations in 
the presence of an external field, using an argument which has some
similarities with our method. 

We consider the following class of quantum spin
systems.
Let $L$ be an even integer and
$\Lambda = \{-\frac12 L, \dots, \frac12L \}^{d} \subset \bbZ^{d}$.
Write $\caE_\L$ for the set of nearest
neighbors in $\L$.  The Hilbert space is 
$\caH_{\Lambda}= \otimes_{x\in\Lambda} \bbC^{2}$ and the Hamiltonian is
\be
\label{def Ham}
H_{\Lambda,\bsh} = -2 \sum_{xy \in \caE_\L} \bigl( S_{x}^{1} S_{y}^{1} +
(2u-1) S_{x}^{2} S_{y}^{2} + S_{x}^{3} S_{y}^{3} - \tfrac14 \bigr)
- \sum_{x\in\Lambda} h_{x} S_{x}^{3}.
\ee
Here, $S_{x}^{i}$ are the usual spin operators that satisfy the
commutation relations $[S_{x}^{1}, S_{y}^{2}] = \ii \delta_{x,y}
S_{x}^{3}$, and further relations obtained by cyclic permutation of
the indices 1, 2, 3.  The parameters $\bsh = (h_{x})_{x\in\Lambda}$ represent external magnetic fields; we assume that they take values in $[0,\infty)$.
The parameter $u$ belongs to $[0,1]$ and
well-known models are obtained for certain values.
The main examples are the $S=\frac12$ Heisenberg and XY models in
transverse fields,   obtained by taking $u=1$ for
the Heisenberg ferromagnet, $u=0$ for
the Heisenberg anti-ferromagnet (up to unitary equivalence),
and $u=\tfrac{1}{2}$  for the XY model.

We actually discuss a more general setting allowing
$S \in \frac12 \bbN$, that is compatible with
the loop-representation (the case $S=\tfrac12$ is
physically the most relevant).  Let $S \in \frac12 \bbN$, and let us consider
the Hilbert space 
$\caH_{\Lambda}= \otimes_{x\in\Lambda} \bbC^{2S+1}$ and the Hamiltonian
\be
\label{def gen Ham}
H_{\Lambda,\bsh} = -\sum_{xy \in \caE_\L} \bigl( uT_{xy} + (1-u) Q_{xy} -
1 \bigr) - \sum_{x\in\Lambda} h_{x} S_{x}^{3}.
\ee
In order to define the operators $T$ and $Q$
that appear above, let $|a\rangle$,
$a = -S, -S+1, \dots, S$ denote a basis of $\bbC^{2S+1}$
of eigenvectors for $S^3_x$.  Then
$S_{x}^{3} |a\rangle_x = a |a\rangle_x$. The transposition operator
$T_{xy}$ acts on $\bbC^{2S+1} \otimes \bbC^{2S+1}$ as follows:
\be
T_{xy} |a\rangle_x \otimes |b\rangle_y = 
|b\rangle_x \otimes |a\rangle_y,
\ee
and the operator $Q_{xy}$ has matrix elements
\be
\langle a|_x \otimes \langle b|_y Q_{x,y} 
|c\rangle_x \otimes |d\rangle_y 
= \delta_{a,b} \delta_{c,d}.
\ee
In the case $S=\frac12$, one can check that the Hamiltonian
of \eqref{def gen Ham} is equal to that of \eqref{def Ham}.

For suitable observables $M$ the finite-volume
states are defined by 
\[
\langle M\rangle_{\L,\bsh}=\frac{\Tr Me^{-\b H_{\L,\bsh}}}{Z(\b,\L,\bsh)},
\quad\mbox{ where } Z(\b,\L,\bsh)=\Tr e^{-\b H_{\L,\bsh}},
\]
and where $\b>0$ denotes the inverse temperature.

Our results consist of two theorems. In Theorem \ref{decay_thm} we assume a uniform lower bound on all $h_{x}$s and we obtain a bound for the transverse correlations that is uniform in the size of the system and in $\beta$:

\begin{theorem}
\label{decay_thm}
Assume that $h_{x} \geq \alpha$ for all $x\in\Lambda$ and some $\alpha>0$. Then there exist constants $C,c>0$ (they depend on $S,d,\alpha$, but not on $L,\beta$) such that
\[
0 < \langle S_{0}^{1} S_{x}^{1} \rangle_{\L,\bsh} < C \e{-c \|x\|}
\]
for all $x\in\Lambda$.
\end{theorem}

Let us remark that similar results follow from the Lee--Yang theorem, as observed by Lebowitz and Penrose \cite{LP}. Let $s\bsh = (sh_{x})$ with $s \in \bbC$. It can be shown that the two point function $\langle S_{0}^{1}
S_{x}^{1} \rangle_{\Lambda,s\bsh}$ is analytic in $s$ when $\Re s \neq 0$. Assume that $\bsh$ is such that the thermodynamic limit $\langle S_{0}^{1}
S_{x}^{1} \rangle_{s\bsh}$ exists. The inverse correlation length
\be
\xi^{-1}(s) = -\limsup_{\|x\|\to\infty} \frac1{\|x\|} 
\log \bigl| \langle S_{0}^{1} S_{x}^{1} \rangle_{s\bsh} \bigr|
\ee
is therefore subharmonic. A cluster expansion shows that
$\xi^{-1}(s) > 0$ when $\Re s$ is large; then it never vanishes in
the domain of analyticity. We refer \cite{FR, LP, Pfi} for more
information.

Our proof of Theorem \ref{decay_thm} is new and
very different. We use the
random-loop-representation of \cite{AN, Toth, Uel1} in order to obtain
a suitable expression for the two-point function. It can be bounded by
the two-point function of a model of dependent percolation. Stochastic
domination allows to remove dependence, and our theorem follows from
results about first-passage percolation. This method of proof seems more robust.

The second result deals with a quenched disordered system where the
$h_{x}$s are i.i.d.\ random variables taking values in
$[0,\infty)$. Let $\bE$ denote expectation with respect to the
magnetic fields $\bsh$. The transverse two-point function is defined as 
\be
\llangle S_{0}^{1} S_{x}^{1} \rrangle_{\Lambda} = \bE \biggl(
\frac1{Z(\beta,\Lambda,\bsh)} \Tr S_{0}^{1} S_{x}^{1} \e{-\beta
H_{\Lambda,\bsh}} \biggr).  
\ee

We allow a small fraction of magnetic fields to be zero. The Lee--Yang method does not seem to apply any more. Our result is not uniform in $\beta$; the situation of the ground state remains to be clarified.

\begin{theorem}
\label{thm quenched decay}
For every $S, d$, $\b$ there exists $\varepsilon>0$ such that, if
$\bbP(h_{x} < \alpha) < \varepsilon$ for some $\alpha>0$, there exist
$C,c>0$ (they depend on $S,d,\beta,\alpha$ but not on $L$) such that
\[
0 < \llangle S_{0}^{1} S_{x}^{1} \rrangle_{\Lambda} \leq C \e{-c \|x\|}
\]
for all $x\in\Lambda$.
\end{theorem}

Straightforward modifications of our argument also gives
the same result for the Schwinger functions
$\langle S_0^1 e^{-(\b-t)H_{\L,\bsh}}S_x^1e^{-t H_{\L,\bsh}}\rangle_{\L,\bsh}$
and 
$\llangle S_0^1 e^{-(\b-t)H_{\L,\bsh}}S_x^1e^{-t H_{\L,\bsh}}\rrangle_{\L}$.
Under the assumptions of Theorem~\ref{decay_thm} 
we can show that for all $t \in [0,\beta]$,
\begin{equation}\label{sch_eq}
\langle S_0^1 e^{-(\b-t)H_\L}S_x^1e^{-t H_{\L,\bsh}}\rangle_{\L,\bsh}
< C\e{-c(\|x\|+t)}.
\end{equation}
The constants $C,c$ are positive and they do not depend on $L,\beta,t,x$.
Under the assumptions of 
Theorem~\ref{thm quenched decay} we obtain a similar upper bound on 
$\llangle S_0^1 e^{-(\b-t)H_\L}S_x^1e^{-t H_{\L,\bsh}}\rrangle_{\L}$
(in this case, the constants are not uniform in $\b$).

We explain the random loop representation in Section \ref{sec random loops} and use it to prove Theorems \ref{decay_thm} and \ref{thm quenched decay} in Section \ref{sec proofs}.

\section{Random loops}
\label{sec random loops}

We now describe an ensemble of random loops.
Its relevance for the spin system
is explained in Theorem~\ref{thm repr}.

The loops live in $\L\times[0,\b)$, and
we  regard the interval $[0,\b)$ as a circle 
of length $\b$.  Points in $\L\times\{0\}$
are identified with the corresponding elements
of $\L$ and denoted $0,x$ et.c.\
We consider
two independent Poisson processes in the set
$\cE_\L\times[0,\b)$.
The first process 
has intensity $u$ and is called the process 
of \emph{crosses};  the second
process has intensity $1-u$
and is called the process of \emph{double bars} 
(or \emph{bars} for short).  The joint 
realization of bars and crosses is denoted by $\om$ and
its distribution is denoted by $\rho$.  
Note that $\om$, taken as a whole,
is a realization of a Poisson process of intensity 1.

The realization $\om$ decomposes $\L\times[0,\b)$ into
a collection of disjoint loops.  Informally, these loops are
obtained as follows.   One starts at 
a point $(x,t)\in\L\times[0,\b)$
and proceeds `upwards' (or `downwards')
until hitting the endpoint of 
a bar or a cross.  One then
 moves to the other endpoint and  
proceeds in the same direction if it was a cross, 
alternatively  changes direction  it was a bar.
The loop is completed when one returns to the starting
point $(x,t)$.  We write $\cL(\om)$ for the collection of
loops defined by $\om$.  For more details, and illustrations, 
see~\cite{Uel1}.

Let us define the relevant loop activities. Given
$\gamma \in \caL(\omega)$, let $\ell_{y}(\gamma)$ denote the
vertical length of $\g$ at the site $y$
(that is, the length of $\g\cap(\{y\}\times[0,\b))$). 
Notice the following identity, that holds for all realizations $\omega$:
\be
\sum_{\gamma \in \caL(\omega)} \sum_{y \in \Lambda} \ell_{y}(\gamma) = \beta |\Lambda|.
\ee
If there is a loop $\gamma_{0,x}\in\cL(\om)$  that
contains both $0$ and $x$, we let $\ell^{+}_{y}(\gamma_{0,x})$ denote
the vertical length at $y$ of the component of the loop that links
$(0, 0+)$ with $(x, 0\pm)$; 
that is, the component obtained by starting
in the `upwards' direction at $(0,0)$ and continuing until
the first visit to $(x,0)$.  We also let
$\ell^{-}_{y}(\gamma_{0,x})$ denote the length at $y$ of the other component
that links $(0, 0-)$ with $(x, 0\pm)$;  note that 
$\ell_{y}(\g_{0,x})= \ell^{+}_{y}(\g_{0,x}) + \ell^{-}_{y}(\g_{0,x})$. 
Define
\be
\begin{split}
& z_\bsh(\gamma) = \sum_{a=-S}^{S} 
\exp\Big( a \sum_{y} h_{y} \ell_{y}(\gamma)\Big), \\ 
& \tilde z_\bsh(\gamma_{0,x})
= \tfrac14 \sum_{a=-S}^{S-1} \bigl( S(S+1) -
a(a+1) \bigr) \biggl[ \exp \Big(
(a+1) \sum_{y} h_{y} \ell^{+}_{y}(\gamma_{0,x}) + a \sum_{y} h_{y} \ell^{-}_{y}(\gamma_{0,x}) \Big) \\
& \hspace{3cm} + \exp \Big(
a \sum_{y} h_{y} \ell^{+}_{y}(\gamma_{0,x}) + (a+1) \sum_{y} h_{y} \ell^{-}_{y}(\gamma_{0,x}) \Big) \biggr].
\end{split}
\ee
(Here and in all similar sums the index $a$ increases in steps 
of size 1.)  We write 
$1_{0 \leftrightarrow x}(\omega)$ for the indicator that
$0$ and $x$ belong to the same loop $\g_{0,x}\in\cL(\om)$.

\begin{theorem}
\label{thm repr}
The partition function and the two-point function have the 
following representations.
\begin{itemize}
\item[(a)] $\displaystyle Z(\beta,\Lambda,\bsh) 
= \int\rho(\dd\omega) \prod_{\gamma \in \caL(\omega)} z_\bsh(\gamma)$.
\item[(b)] $\displaystyle \langle S_{0}^{1} S_{x}^{1} \rangle_{\L,\bsh} 
= \frac1{Z(\beta,\Lambda,\bsh)} \int\rho(\dd\omega) \; 
1_{0 \leftrightarrow x}(\omega) \; \tilde z_\bsh(\gamma_{0x}) 
\prod_{\gamma \in \caL(\omega) \setminus \{\gamma_{0x}\}} z_\bsh(\gamma)$.
\end{itemize}
\end{theorem}

This theorem builds on~\cite[Theorems~3.2 and~3.3]{Uel1},
and is proved at the end of this section. Notice that (b)
shows that the two-point function is positive. 

The following corollary
`essentially' shows exponential decay --- but it takes a surprising
effort in order to turn it into a rigorous proof.
We write $\EE_\bsh$ for expectation with respect to the probability 
measure with density proportional
to $\prod_{\gamma \in \caL(\omega)} z_\bsh(\gamma)$ with respect
to $\rho$.

\begin{corollary}\label{loop_cor}
We have the estimate
\[
\langle S_{0}^{1} S_{x}^{1} \rangle_{\L,\bsh} 
\leq \tfrac13 S (S+1) (2S+1) \; 
\bbE_\bsh\Bigl( \; 1_{0 \leftrightarrow x}(\omega) \; 
\e{-\sum_{y} h_{y} \ell^{+}_{y}(\gamma_{0,x})} \Bigr).
\]
\end{corollary}

\begin{proof}
By Theorem~\ref{thm repr} we have 
$\langle S_{0}^{1} S_{x}^{1} \rangle_{\Lambda,\bsh} = 
\bbE_\bsh\bigl( 1_{0 \leftrightarrow x}(\omega) 
\tilde z_\bsh(\gamma_{0,x}) / z_\bsh(\gamma_{0,x}) \bigr)$.
The loop activity $z_{\bsh}$ satisfies the lower bound
\be
z_\bsh(\gamma) \geq \e{S \sum_{y} h_{y} \ell_{y}(\gamma)}.
\ee
As for $\tilde z_\bsh$, we have the upper bound
\bm
\tilde z_\bsh(\g_{0,x}) \leq \tfrac14 
\Bigl( \e{ \sum_{y} h_{y} ( S \ell^{+}_{y}(\g_{0,x}) + (S-1) \ell^{-}_{y}(\g_{0,x}))} 
+ \e{ \sum_{y} h_{y} ( (S-1) \ell^{+}_{y}(\g_{0,x}) + S \ell^{-}_{y}(\g_{0,x}))} \Bigr) \\
\cdot \sum_{a=-S}^{S-1} \bigl( (S(S+1) - a(a+1) \bigr).
\end{multline}
One can check that the latter sum is equal to $\frac23 S (S+1) (2S+1)$. 
Thus
\be
\tilde z_\bsh(\g_{0,x}) \leq \tfrac16 S (S+1) (2S+1) 
\e{S \sum_{y} h_{y} \ell_{y}(\g_{0,x})} 
\Bigr( \e{- \sum_{y} h_{y} \ell^{+}_{y}(\g_{0,x})} + 
\e{- \sum_{y} h_{y} \ell^{-}_{y}(\g_{0,x})} \Bigr).
\ee
The corollary follows.
\end{proof}

\begin{proof}[Proof of Theorem \ref{thm repr}]
By the Trotter product formula, with $\rho$ the Poisson 
point process described above, we have
\be
\begin{split}
&\e{-\beta H_{\Lambda,\bsh}} = 
\lim_{N\to\infty} \biggl[ \Bigl( 1 - \tfrac1N |\caE| + 
\frac1N \sum_{xy \in \caE} \bigl( u T_{xy} + (1-u) Q_{xy} \bigr) \Bigr) 
\e{\frac1N \sum_{x} h_{x} S_{x}^{3}} \biggr]^{\beta N} \\
&= \int\rho(\dd\omega) 
\e{(\beta-t_{n}) \sum_x h_{x} S_{x}^{3}} R_{x_{n} y_{n}} 
\e{(t_{n}-t_{n-1}) \sum_x h_{x} S_{x}^{3}} \dots 
\e{(t_{2}-t_{1}) \sum_x h_{x} S_{x}^{3}} R_{x_{1} y_{1}} 
\e{t_{1} \sum_x h_{x} S_{x}^{3}}.
\end{split}
\ee
Here, $(t_{i}; x_{i}, y_{i})$ are the times and locations of the
outcomes of the realization $\omega$, ordered so that 
$0 < t_{1}< \dots < t_{n} < \beta$. 
The operator $R_{x_{i} y_{i}}$ is equal to
$T_{x_{i} y_{i}}$ if the outcome $(t_{i}; x_{i}, y_{i})$ is a cross;
it is equal to $Q_{x_{i} y_{i}}$ if the outcome is a double bar.

Inserting the expansion of unity $\bbone
= \sum_{(\sigma_{x})} \otimes_{x}
|\sigma_{x}\rangle \langle \sigma_{x}|$, with
$\sigma_{x} \in \{-S,\dots,S\}$, one obtains
\be
Z(\beta,\Lambda,\bsh) = \Tr \e{-\beta H_{\Lambda,\bsh}} 
= \int\rho(\dd\omega) \sum_{\sigma \in \Sigma(\omega)} 
\exp\Bigl\{ \int_{0}^{\beta} \dd t \sum_{x} h_{x} \sigma_{x}(t) \Bigr\}.
\ee
The last sum is over `space-time spin configurations' 
$\sigma=(\s_x:x\in\L)$ that are compatible with $\omega$. 
That is, $\sigma$ is a (periodic) function 
$[0,\beta) \to \{-S,-S+1,\dots,S\}^{\Lambda}$,
and satisfies
\begin{itemize}
\item $\sigma(t)$ is constant except possibly 
at times $t_{1}, \dots, t_{n}$;
\item At those times, we have 
$\langle \sigma(t_{i}+)| R_{x_{i} y_{i}} |\sigma(t_{i}-)\rangle = 1$.
\end{itemize}
It is not hard to check that $\sigma \in \Sigma(\omega)$ if and only
if the spin values are constant on each loop of $\caL(\omega)$. The
sum over space-time spin configurations factorize according to the
loops, and we get the claim (a) of Theorem \ref{thm repr}.

For the correlation function, let $\Sigma_{0,x}(\omega)$ be the set of
space-time configurations $[0,\beta) \to \{-S,\dots,S\}^{\Lambda}$
that satisfy
\begin{itemize}
\item $\sigma(t)$ is constant except possibly 
at times $0,t_{1}, \dots, t_{n}$;
\item At times $t_{1}, \dots, t_{n}$, we have 
$\langle \sigma(t_{i}+)| R_{x_{i} y_{i}} |\sigma(t_{i}-)\rangle = 1$;
\item At time 0, we have
\be
\sigma_{y}(0+) = \begin{cases} 
\sigma_{y}(0-), & \text{if } y \neq 0,x; \\ 
\sigma_{y}(0-) \pm 1, & \text{if } y = 0 \text{ or } x. 
\end{cases}
\ee
\end{itemize}
Since $\langle a| S^{1} |b\rangle = \frac12 \sqrt{S(S+1) - ab}$ if $a
= b\pm1$, and 0 otherwise, we have
\be
\Tr S_{0}^{1} S_{x}^{1} \e{-\beta H_{\Lambda,\bsh}} 
= \tfrac14 \int\rho(\dd\omega) 1_{0 \leftrightarrow x}(\omega) 
\sum_{\sigma \in \Sigma_{0.x}(\omega)} 
\bigl( S(S+1) - \sigma_{0}(0-) \sigma_{0}(0+) \bigr) 
\e{\int_{0}^{\beta} \dd t \sum_{y} h_{y} \sigma_{y}(t)}.
\ee
Extracting the contribution of the loop $\gamma_{0,x}$, and using the
definition of $\tilde z_\bsh(\gamma_{0,x})$, one obtains 
Theorem~\ref{thm repr} (b).
\end{proof}

The step from Corollary~\ref{loop_cor} to
Theorems~\ref{decay_thm} and~\ref{thm quenched decay} 
is intuitively  clear:
on the event  $0\lra x$ we expect $\g_{0,x}$
to have length proportional to $\|x\|$, so the
right-hand-side in   Corollary~\ref{loop_cor} 
should decay exponentially in $\|x\|$.
The difficulty is that $\ell_y^+(\g_{0,x})$ denotes 
\emph{vertical} length.  For any $\eps>0$ it 
is possible for
a loop to reach from 0 to $x$, yet still
have total vertical length at most $\eps$.  This
seems unlikely when $\eps$ is small, 
but obtaining a quantitative
statement requires dealing with
the dependencies under $\EE_\bsh$.

\section{Proofs}
\label{sec proofs}

We begin by noting that both 
$\langle S_0^1S_x^1\rangle_{\L,\bsh}$ and 
$\llangle S_0^1S_x^1\rrangle_{\L}$
can be written in the
 general form
$\bE(\langle S_0^1S_x^1\rangle_{\L,\bsh})$, 
where now $\bE$ is a measure governing the vector 
$\bsh$ under which the
$h_x$ are independent (but not necessarily identically distributed).  
Indeed,  $\llangle S_0^1S_x^1\rrangle_{\L}$
is obtained by letting the $h_x$ be identically distributed, whereas 
$\langle S_0^1S_x^1\rangle_{\L,\bsh}$ is obtained when $\bE$ is the
degenerate measure under which the $h_x$ are almost surely 
constant.  In
either case, by Corollary~\ref{loop_cor} the two-point function
is bounded by a constant times
\be\label{2p-exp}
\bE\Bigl(\bbE_\bsh\Bigl( \; 1_{0 \leftrightarrow x}(\omega) \; 
\e{-\sum_{y} h_{y} \ell^{+}_{y}(\gamma_{0,x})} \Bigr)\Bigr).
\ee
We focus on bounding the quantity~\eqref{2p-exp}, and at the end
deduce Theorems~\ref{decay_thm} 
and~\ref{thm quenched decay} by specializing to the specific
choices for $\bE$.

Let $\d>0$ be such that that $N=\b/\d$ is an integer.  In what follows
we no longer need to distinguish between bars and crosses,
and we use the term \emph{bridges} to refer collectively to the two.  
We write $\G=\G(\L,\b,\d)$ for the collection of intervals
of the form
\[
I=\{x\}\times [k\d,(k+1)\d),\mbox{ for }
x\in\L\mbox{ and } 0\leq k\leq N-1.
\]
We view $\G$ as a graph, where intervals
$\{x\}\times [k\d,(k+1)\d)$ and
$\{y\}\times [\ell\d,(\ell+1)\d)$ are said to be \emph{adjacent} if
either 
\begin{enumerate}
\item $xy\in\cE_\L$ and $k=\ell$, or
\item $x=y$ and $k=\ell\pm 1$ (viewed modulo $N$).
\end{enumerate}
A \emph{path} $\pi$
in $\G$ is as usual a sequence of elements of $\G$ which
are consecutively adjacent in this sense, and any such path thus
corresponds to a sequence of `neighbouring' intervals.  

Fix $\a>0$, to be chosen later, and
let $I=\{x\}\times[k\d,(k+1)\d)\in\G$. 
Based on the random outcomes $\bsh$ and
$\om$ (i.e.\ the collection of bridges) we
will declare the interval $I$ to be
\begin{itemize}
\item $\bsh$-good if $h_x\geq\a$,
\item $\om$-good if there is no bridge with an endpoint in $I$, 
\item and \emph{good} if it is both $\bsh$-good and $\om$-good. 
\end{itemize} 
An interval which is not declared good is declared \emph{bad}.
We encode the collection of good and bad intervals as an element 
$\eta=(\eta(i):i\in\G)$ of $\{0,1\}^\G$, where 0 denotes bad and 1
denotes good.  This classification may be seen as a (dependent)
percolation process in $\G$.

It is convenient to use the fact that $\ZZ^d$ is bipartite:
We may write $\ZZ^d=\cA\cup\cB$ where
\[
\cA=\{(x_1,\dotsc,x_d)\in\ZZ^d:
x_1+\dotsb+x_d\equiv0\mbox{ (mod 2)}\}
\]
and $\cB=\ZZ^d\sm\cA$.  We refer to $\cA$
and $\cB$ as the even and odd sublattices,
respectively.  Bipartiteness refers to the fact that
a vertex in $\cA$ is only adjacent to 
vertices in $\cB$, and vice versa.
If $I=\{x\}\times [k\d,(k+1)\d)$ is an interval belonging 
to $\G$ we commit a small abuse of notation and write
$I\in\cA$ if $x\in\cA$.
We also simply write $x$ for the unique interval
$\{x\}\times[0,\d)$ of $\G$ containing $(x,0)$. 

We define the \emph{passage time} $T_\L(x)$ from $0$ to $x$ 
in $\L$ as
\begin{equation}\label{passage_eq}
T_\L(x)=\min_{\pi:0\to x} \sum_{i\in\pi\cap\cA} \eta(i),
\end{equation}
where the minimum is over all paths $\pi$
in $\G$ from 
$\{0\}\times[0,\d)$ to $\{x\}\times[0,\d)$.
Thus $T_\L(x)$ is the minimal number of good intervals,
indexed by the even sublattice,
on a path from 0 to $x$.   This is a slight variation of the standard
definition of a (point-to-point) passage time, where the sum would
usually go over \emph{all} points on $\pi$.  Summing over the even
sublattice $\cA$ only is a convenient way 
to avoid dependencies, as will be
explained below.

Let $\varphi>0$ be arbitrary, and assume that
the event $0\lra x$ occurs (so $\g_{0,x}$ is well-defined).  
If, in addition, $T_\L(x)\geq\varphi\|x\|$, then any path in $\G$
from 0 to $x$ contains at least $\varphi\|x\|$ good intervals.
In particular, 
it follows that  $\sum_y\ell_y^+(\g_{0,x})\geq \a\d(\varphi\|x\|-1)$.
(We subtract 1 for the last interval, which may contribute
less than $\a\d$ even if it is good.)
Thus we have that 
\begin{equation}\label{cases_eq}
\bE\Bigl(\bbE_\bsh\Bigl( \; 1_{0 \leftrightarrow x}(\omega) \; 
\e{-\sum_{y} h_{y} \ell^{+}_{y}(\gamma_{0,x})} \Bigr)\Bigr)
\leq
e^{-\a\d (\varphi\|x\|-1)}+\bE[\PP_\bsh(T_\L(x)<\varphi\|x\|)].
\end{equation}
The  theorems follow if we show that the
last term is exponentially small in $\|x\|$.

To establish this, we first simplify the
probability measure by using the theory of stochastic
domination.  We begin by defining a partial ordering.
We say that $\om\leq \tilde\om$ if one may 
obtain $\om$ from $\tilde\om$
by removing some bridges (in other words, the support of $\om$ is a
subset of the support of $\tilde\om$).
For any fixed $\bsh$, the event
$A=\{T_\L(x)<\varphi\|x\|\}$ is \emph{increasing}
in the ordering on $\om$, i.e.\ if $\om\leq\tilde\om$
and $\om\in A$ then necessarily $\tilde\om\in A$.
This allows us to use results on stochastic domination for
point processes~\cite{georgii-kuneth,preston},
to bound $\PP_\bsh(A)$.

The following result lets us get rid of the complicated density 
$\prod_{\g\in\cL(\om)}z_\bsh(\g)$ at the cost of increasing the intensity
of bridges.  Write $\theta=2S+1$.
\begin{lemma}\label{preston_lem}
Let $\PP'$ denote the probability measure under which the bridges
form a Poisson process of intensity $\theta$.
Then for any realization of $\bsh$ we have that
\[
\PP_\bsh(T_\L(x)<\varphi\|x\|)\leq \PP'(T_\L(x)<\varphi\|x\|).
\]
\end{lemma}
Before turning to the proof, we note that $\PP'$ may
alternatively be described as follows.  For each  
pair $xy\in\cE_\L$,
the process of bridges on $\{xy\}\times[0,\b)$ is 
(under $\PP'$) a Poisson  process with intensity $\theta$.
 For all other pairs $x'y'\in\cE_\L$
the processes of bridges on $\{xy\}\times[0,\b)$
and on $\{x'y'\}\times[0,\b)$ are independent.  
\begin{proof}
If $\g_1,\g_2$ are disjoint, measurable
subsets of $\L\times[0,\b)$, a calculation shows that
\begin{equation}\label{add_bridge_eq}
\frac{1}{\theta}\leq
\frac{z_\bsh(\g_1\cup\g_2)}{z_\bsh(\g_1)z_\bsh(\g_2)}
\leq 1\leq \theta.
\end{equation}
If $\tilde\om$ is obtained from $\om$ by adding a bridge then either
some loop $\g\in\cL(\om)$ decomposes into two loops $\g_1,\g_2$,
or two loops $\g_1,\g_2\in\cL(\om)$ are joined to form a
larger loop $\g$, or the loops stay the same.  
In either case,~\eqref{add_bridge_eq} shows that
\[
\frac{\prod_{\tilde\g\in\cL(\tilde\om)}z_\bsh(\tilde\g)}
{\prod_{\g\in\cL(\om)}z_\bsh(\g)}\leq \theta
=\frac{\theta^{|\tilde\om|}}{\theta^{|\om|}}.
\]
It follows from~\cite[Theorem~1.1]{georgii-kuneth} 
that, for any event $A$ 
which is increasing in the
ordering on $\om$, the probability $\PP_\bsh(A)$ is dominated by the
probability of $A$ under the measure with density 
proportional to $\theta^{|\om|}$
with respect to $\rho$.  The latter measure is precisely $\PP'$,
so the result follows.
%
\end{proof}

In~\eqref{cases_eq} we need to bound
\be
\bE[\PP_\bsh(T_\L(x)<\varphi\|x\|)\leq
\bE[\PP'(T_\L(x)<\varphi\|x\|)=
\EE'[\bP(T_\L(x)<\varphi\|x\|)].
\ee
Let $\bP'$ (or $\bE'$) denote the measure under which the classification
of the intervals $I=\{x\}\times[k\d,(k+1)\d)\in\G$ 
into $\bsh$-good and $\bsh$-bad
is done independently over all such intervals, 
with probability $\bP(h_x<\a)^{1/N}$ for $\bsh$-bad.  Then a
straightforward coupling shows that, for any fixed $\om$,
\be
\bP(T_\L(x)<\varphi\|x\|)\leq\bP'(T_\L(x)<\varphi\|x\|).
\ee
(For any $x$ the $\bP'$-probability that \emph{all} intervals 
$I=\{x\}\times[k\d,(k+1)\d)$ are $\bsh$-bad is exactly the same as under
$\bP$, and for all other outcomes the $\bP'$-realization has more bad
intervals than the $\bP$-realization.)  It follows that
\be
\bE[\PP_\bsh(T_\L(x)<\varphi\|x\|)\leq
\bE'[\PP'(T_\L(x)<\varphi\|x\|).
\ee

Under $\PP'$ the labels $\om$-good and $\om$-bad 
assigned to the intervals
in $\G$ are almost independent.   In fact, they are independent
if the intervals are at vertical distance at least 1, or at horizontal
distance at least 2.  This is because the labels assigned to such
intervals are functions of the realization of the Poisson process 
$\om$ in disjoint intervals, and are therefore independent.  Thus,
since we look only at the even sublattice $\cA$,
under $\bE'\times\PP'$
 the labels $\eta(i)$ assigned to the vertices in the
sum in~\eqref{passage_eq} are independent.  

At this point we comment on  differences between 
the two theorems and on the choice of $\a$.
Each $i\in\G$ is `good' with probability at least 
$p:=(1-P(h_x<\a)^{1/N})e^{-2d\theta\d}$.  In the case of 
Theorem~\ref{decay_thm} we take $\a$ as in the statement of the
result, so that all the $h_x$ are
uniformly bounded from below by $\a$.  Then 
$p=e^{-2d\theta\d}$, which firstly does not depend on $x$, and 
secondly approaches 1 uniformly in $\b$ as $\d\to0$.
In the case of Theorem~\ref{thm quenched decay}, 
the $h_x$ are identically distributed, 
so again $p$ does not depend on $x$ (for any $\a$).  
In this case we need to
 pick $\d>0$ small enough and then $\a>0$ small
enough to make $p$ close to 1.

The next step will be to use a general result from the theory of
first-passage percolation. 
Recall that $\cA$ is the even sublattice of $\ZZ^d$, and 
let 
\[
\Xi=\cA\times\{0,\dotsc,N-1\}.
\]   
We view $\Xi$ 
as a graph as follows.  If $x,y\in\cA$ and 
$k,\ell\in\{0,\dotsc,N-1\}$ then $(x,k)$ and $(y,\ell)$ are adjacent in
$\Xi$ if either $x=y$ and $k=\ell\pm 1$ (mod $N$), or if
$k=\ell$ and $x$ and $y$ are \emph{next}-nearest 
neighbours in $\ZZ^d$.  
Thus  $\Xi$ inherits the natural adjacency relation in $\cA$
for the first coordinate, and is also
`periodic' in the last coordinate.  

For convenience we introduce the
probability measure  $\tilde\PP$
which assigns values 0 or 1 to the
elements of $\Xi$ independently, with probability 
$p$ (defined above) for 1.
We also define the
infinite-volume passage time $T(x)$ in the same way as
in~\eqref{passage_eq}, except that we replace the minimum by an 
infimum taken over \emph{all} paths in $\Xi$.  
Note that $T(x)\leq T_\L(x)$, so by 
the reasoning above we have in~\eqref{cases_eq} that
\begin{equation}\label{iid_eq}
\bE[\PP(T_\L(x)<\varphi\|x\|)]\leq \tilde\PP(T(x)<\varphi\|x\|).
\end{equation}
The theorems will follow by applying the following lemma.

\begin{lemma}\label{fpp_lem}
There is $\kappa>0$, depending only on $d$, such that if $p>1-\kappa$
then there are constants $\varphi>0$ and $c_1,c_2>0$,
depending only on $p$ and $d$, such that
\begin{equation}\label{fpp_expdecay}
\tilde\PP(T(x)<\varphi\|x\|)\leq c_1 e^{-c_2\|x\|}.
\end{equation}
\end{lemma}

In fact, here it suffices if $\kappa\leq(2d)^{-2}$.
(The relevance of the number $(2d)^{-2}$ is that it is a lower bound
on the critical probability for site percolation on $\Xi$, as can
easily be proved using standard path-counting arguments, 
see e.g.~\cite[Theorem~1.10]{grimmett-perc}.)
Theorem~\ref{decay_thm} follows on taking $\d>0$ sufficiently small:
as noted above $p$ is then close to 1 uniformly in $\b$ so all the
constants in~\eqref{fpp_expdecay} are uniform in both $\L$ and $\b$.
For Theorem~\ref{thm quenched decay} we first take $\d>0$ 
small, and find that there is $\eps>0$ such that if 
$\bP(h_x<\a)<\eps$ then $p>1-\kappa$, but $\eps$ will now depend
on $N$ and hence $\b$.   (For an explicit bound,
$\eps\leq (1-(1-\kappa)e^{2d\theta\d})^N$ suffices.)

\begin{proof}[Sketch proof of Lemma~\ref{fpp_lem}]
This can be proved by 
adapting~\cite[Proposition~5.8]{Kesten-FPP} 
(see also~\cite{gr-kest-84}).  That result deals with
bond-first-passage-percolation on $\ZZ^d$, and 
our situation is slightly
different since we are dealing with \emph{site}-percolation on the 
sublattice $\cA$ with the next-nearest-neighbour adjacency relation, 
and also the underlying graph is periodic in one
direction.  We give a rough outline of the main ideas.

Write $n=\|x\|$.  On the event that $T(x)<\a n$
there must be a self-avoiding walk $w$ in $\Xi$ which starts at the origin
and contains at least $n$ steps, such that the passage time along
$w$ satisfies
\[
\sum_{i\in w}\eta(i)<\a n.
\]
One may decompose $w$ into a finite 
sequence $w_1,w_2,\dotsc$ of sub-walks, each of which traverses
distance $m$ for some fixed $m$, and the sum of whose passage times is 
`small'.  Since these paths are disjoint, we obtain an upper
bound if we assume that the corresponding passage times are
independent, by the \textsc{bk}-inequality~\cite{bk}.
Since $1-p$ is subcritical for site-percolation in $\Xi$, 
the set of vertices
with passage time 0 from the origin does not percolate, meaning that
for suitable $m$ the passage times for distance $m$ are very unlikely
to be small, and exponential decay follows from a
large-deviations-type estimate.
\end{proof}

For the extension~\eqref{sch_eq} to Schwinger functions, 
we note that
small modifications of Theorem~\ref{thm repr} and
Corollary~\ref{loop_cor} give that
\[
\langle S_0^1 e^{-(\b-t)H_{\L,\bsh}}S_x^1e^{-t H_{\L,\bsh}}\rangle_{\L,\bsh}
\leq \tfrac13 S(S+1)(2S+1)
\EE_\bsh\big(1_{0\lra (x,t)}\e{-\sum_y h_y \ell_y^+(\g_{0,x})}\big).
\]
Here $0\lra(x,t)$ denotes the event that $(0,0)$ and $(x,t)$
lie in the same loop, $\g_{0,x}$ is that loop, 
and $\ell_y^+(\g_{0,x})$ is defined as 
before except that 
it concerns the part of the loop up to $(x,t)$ rather than to $(x,0)$.
The result then follows from straightforward adjustments of the
definition~\eqref{passage_eq} of the passage time $T_\L$ as well as of
the remaining arguments in this section.

\subsection*{Remark}
The key inequality~\eqref{iid_eq} can also be obtained by an appeal to
the main result of~\cite{lss}.  Indeed, as remarked below
Lemma~\ref{preston_lem} the good/bad labelling $\eta$ forms a
1-dependent random field under $\bE'\times\PP'$, with marginal density at
least $p$ 
(using the terminology of~\cite{lss}).  Hence there
is $q>0$, satisfying $q\to1$ as $p\to1$, such that $\eta$
stochastically dominates an i.i.d.\ field with marginal density
$q$.  (A result of this form can alternatively be obtained by
`hands-on' methods.)  

With this approach it is no longer necessary to
restrict the sum in~\eqref{passage_eq} to the even sublattice $\cA$.
Together with straightforward adaptations of the remaining arguments, 
this allows us to extend the results of this
paper to arbitrary translation-invariant lattices with
uniformly bounded degrees (even if they are not bipartite).

{
\renewcommand{\refname}{\small References}
\bibliographystyle{symposium}

}

\end{document}